\documentclass{llncs}
\usepackage{tikz}
\usepackage{framed}
\usepackage{stmaryrd}
\usepackage{amsmath,amssymb,amsfonts}
\usepackage{varwidth}
\usepackage{xspace}
\usepackage{enumitem}
\usepackage{graphicx}
\usepackage{color}
\usepackage{booktabs}
\usepackage{xcolor,colortbl}
\usepackage{multirow}
\usepackage{url}

\usepackage{cite}

\makeatletter
\newif\if@restonecol
\makeatother

\usepackage{subfig}
\usepackage{fancyhdr}
\usepackage{wrapfig}
\usepackage[ruled,vlined,linesnumbered]{algorithm2e}

\makeatletter
\def\url@mystyle{%
  \@ifundefined{selectfont}{\def\UrlFont{\mathtt}}{\def\UrlFont{}}}
\makeatother
\begin{document}
\title{Joint Probability Distribution of Prediction Errors of ARIMA}

% \author{Xin Qin}
% \affiliation{University of Southern California}
% \email{xinqin@usc.edu}
%
% \author{Jyotirmoy V. Deshmukh}
% \affiliation{University of Southern California}
% \email{jyotirmoy.deshmukh@usc.edu}

\author{Xin Qin , Jyotirmoy V. Deshmukh  \\ Email: $\mathtt{\{xinqin,jyotirmoy.deshmukh\}@usc.edu}$}

\institute{University of Southern
California}
\maketitle

\begin{abstract}
Producing probabilistic guarantee for several steps of a predicted signal follow a temporal logic defined behavior has its rising importance in monitoring. In this paper, we derive a method to compute the joint probability distribution of prediction errors of multiple steps based on Autoregressive Integrated Moving Average(ARIMA) model. We cover scenarios in stationary process and intrinsically stationary process for univariate and multivariate.

\end{abstract}


\section{Introduction}
Signal Temporal Logic(STL) specify the simultaneous behavior of a signal across different time point, while time series prediction based on Autoregressive Integrated Moving Average (ARIMA) model only give guarantee on each single time point. In this paper, we give the proof of calculating joint probability distribution of prediction errors of multiple steps.

\section{Error of Prediction in Univariate Process}
The goal is to get the joint distribution of
$$
\{Err(X_{n+1}),Err(X_{n+2}),\dots ,Err(X_{n+h})\}
$$
\subsection{Error of Stationary Process}
% which have been mentioned in Theorem \ref{EerrorZero} and Theorem \ref{Varerror}:
Previous work \cite{StatTexBook2} \cite{StatTexBook} concluded that for stationary process
\begin{equation*}
\begin{split}
E[Err(X_{n+h})]&=0 \\
E[Err(X_{n+h})]^2&=\gamma_X(0)-(a_n^h)^\top\gamma_n(h)
\end{split}
\end{equation*}

Assuming that $Err(X_{n+h})$ is a normal distribution, we have $Err(X_{n+h})\sim N(0,\gamma_X(0)-(a_n^h)^\top\gamma_n(h))$.

\begin{definition}The best linear prediction of $X_{n+h}$ is\cite{StatTexBook}:
\begin{equation}
P(X_{n+h}\mid X_n,X_{n-1},\dots ,X_1,1)=a_0^h+\sum_{i=1}^na_i^hX_{n+1-i}
\end{equation}

The optimized choice of $a_i^h$ are
\begin{equation}\label{bestchoice}
\begin{pmatrix}
Cov(X_1,X_1) & Cov(X_1,X_2) & \dots & Cov(X_1,X_n)\\
Cov(X_2,X_1) & Cov(X_2,X_2) & \dots & Cov(X_2,X_n)\\
\vdots & \vdots & \vdots & \vdots\\
Cov(X_n,X_1) & Cov(X_n,X_2) & \dots & Cov(X_n,X_n)
\end{pmatrix}
\begin{pmatrix}
a_1 \\ a_2 \\ \vdots\\ a_n
\end{pmatrix}=
\begin{pmatrix}
\gamma_X(h) \\ \gamma_X(h+1) \\ \vdots \\ \gamma_X(n+h-1)
\end{pmatrix}
\end{equation}
and
\begin{equation}
a_0=\mu_X(1-\sum_{i=1}^na_i)
\end{equation}
where $\gamma_X(h)=Cov(X_i,X_{i+h})$
\end{definition}

\begin{definition}\label{ErrorDef}
\begin{equation}
		Err(X_{n+h})=P_nX_{n+h}-X_{n+h}
\end{equation}
\end{definition}

\begin{definition}\label{PredictorPropertyAdd}
For best linear predictor \cite{StatTexBook}
%using property 4\cite{introduction to time series and forecast} of best linear predictor %%%%(cite introduction to time series and forecast page59)
	\begin{equation*}
P_n(\alpha_1U+\alpha_2V+\beta)=\alpha_1P_n(U)+\alpha_2P_n(V)+\beta
\end{equation*}

\end{definition}

\begin{corollary}\label{ObservedEReal}
The value of best linear predictor at already observed datapoint $\{X_1,X_2,\dots,X_n\}$ follows
\begin{equation*}
P_n(\sum_{i=1}^n\alpha_iX_i+\beta)=\sum_{i=1}^n\alpha_iX_i+\beta
\end{equation*}

\end{corollary}

\begin{definition}\label{lineartransform}
If the distribution of $X=\begin{pmatrix}X_1 \\ X_2 \\ \vdots \\X_n \end{pmatrix}$ is a multivariate normal distribution, then the random variable vector $Y=CX+b$, where $C$ is a $r\times n$ matrix, is also a multivariate normal distribution.And if $\{X_t\}\sim N_n(\mu,\Sigma)$, then $\{Y_t\}\sim N_r(C\mu+b,C\Sigma C^\top)$
\end{definition}

\begin{definition}
	
	If $\{Z_t\}\sim WN(0,\sigma^2)$ are independent white noise, the joint distribution of
\begin{equation*}
\begin{pmatrix}
Z_{n-q+1}\\
Z_{n-q+2}\\
\vdots\\
Z_{n+h}
\end{pmatrix}
\end{equation*}
is a multivariate normal distribution with
\begin{equation}
%\begin{array}{c}
\mu_Z=\begin{pmatrix}
0\\
0\\
\vdots\\
0\end{pmatrix}\quad
\Sigma_Z=\begin{pmatrix}
\sigma^2 & 0 & \dots & 0\\
0 & \sigma^2 & \dots & 0\\
\vdots & \vdots & \ddots & 0\\
0 & 0 & \dots & \sigma^2
\end{pmatrix}
%\end{array}
\end{equation}

\end{definition}

%\begin{equation}
%Err(\mathbf{X}_h)=
%\begin{pmatrix}
%Err(X_{T+1})\\
%Err(X_{T+2})\\
%\vdots\\
%Err(X_{T+h})
%\end{pmatrix}
%\end{equation}

\begin{lemma}\label{C1}
Given $Y_t=\theta(B)Z_t\quad {t\geq p+1}$,
let $\mathbf{Y}_{h}$ denote  $[Y_{n + i}]_{i = 1}^h$, then
	\begin{equation}\label{Z2Y}
\mathbf{Y}_h=C_1\begin{pmatrix}
Z_{n-q+1}\\
Z_{n-q+2}\\
\vdots\\
Z_{n+h}
\end{pmatrix}
\end{equation}
where
$C_1$ is a $[h\times(q+h)]$ transform matrix from $\begin{pmatrix}
Z_{n-q+1}\\
Z_{n-q+2}\\
\vdots\\
Z_{n+h}
\end{pmatrix}$ to $\mathbf{Y}_h$.
\begin{equation}\label{Z2Ytransformmatrix}
C_1=\begin{pmatrix}
\theta_q & \theta_{q-1} & \dots & \theta_1 & 1 & 0 & 0 & \dots & 0\\
0 & \theta_q & \theta_{q-1} & \dots & \theta_1 & 1 & 0 & \dots & 0\\
0 & 0 & \theta_q & \theta_{q-1} & \dots & \theta_1 & 1 & \dots & 0\\
\vdots & \vdots & \vdots & \ddots & \ddots & \ddots & \ddots & \ddots & \vdots\\
0 & 0 & 0 & \dots & \theta_q & \theta_{q-1} & \dots & \theta_1 & 1
\end{pmatrix}
\end{equation}
\end{lemma}
\begin{proof}

From the definition of ARMA(p,q) process,
\begin{equation}
\begin{split}
	Y_t&=\theta(B)Z_t\quad {t\geq p+1} \\
		&=Z_t+\theta_1Z_{t-1}+\theta_2Z_{t-2}+\dots +\theta_qZ_{t-q}
\end{split}
\end{equation}
where $Z_t \sim WN(0,\sigma^2)$.

\end{proof}

\begin{lemma}\label{C3C2}
	\begin{equation}
P_n\mathbf{Y}_{n+h}=C_3C_2\begin{pmatrix}
Z_{p-q+1}\\
Z_{p-q+2}\\
\vdots\\
Z_{n}
\end{pmatrix}
\end{equation}
where $C_2$ is a $[(n-p)\times (n-p+q)]$ matrix, $C_3$ is a $[h \times h]$ matrix.
\begin{equation}
C_2=\begin{pmatrix}
\theta_q & \theta_{q-1} & \dots & \theta_1 & 1 & 0 & 0 & \dots & 0\\
0 & \theta_q & \theta_{q-1} & \dots & \theta_1 & 1 & 0 & \dots & 0\\
0 & 0 & \theta_q & \theta_{q-1} & \dots & \theta_1 & 1 & \dots & 0\\
\vdots & \vdots & \vdots & \ddots & \ddots & \ddots & \ddots & \ddots & \vdots\\
0 & 0 & 0 & \dots & \theta_q & \theta_{q-1} & \dots & \theta_1 & 1
\end{pmatrix}
\end{equation}
\begin{equation}
C_3=\begin{pmatrix}
a_{n-p}^1 & a_{n-p-1}^1 & \dots & a_1^1 \\
a_{n-p}^2 & a_{n-p-1}^2 & \dots & a_1^2 \\
\vdots & \vdots & \ddots & \vdots \\
a_{n-p}^h & a_{n-p-1}^h & \dots & a_1^h
\end{pmatrix}
\end{equation}
\end{lemma}

\begin{proof}
	$\{Y_t\}$ is a stationary process starting from $p+1$ with zero mean, the prediction of $\{Y_t\}$ follows
\begin{equation}
P_nY_{n+h}=a_0^h+\sum_{i=1}^{n-p}a_i^hY_{n+1-i}
\end{equation}
where $a_0^h=0$. C3
is the coefficient matrix of $(Y_{p+1},Y_{p+2},\dots,Y_n)$ when calculating $P_n\mathbf{Y}_h$, the best linear predictor of $\{Y_{n+1},Y_{n+2},\dots,Y_{n+h}\}$, a stationary MA(q) process with zero mean.
\end{proof}

\begin{theorem}
	The covariance matrix of prediction error of a moving average process $Y$ is
\begin{equation}
\Sigma_{ErrY}=C_{ZtoErrY}\Sigma_ZC_{ZtoErrY}^\top
\end{equation}
\end{theorem}
\begin{proof}
%Let $\mathbf{X}_{h}$ denote $\begin{pmatrix} X_{T+1}\\ X_{T+2}\\ \vdots\\ X_{T+h}\end{pmatrix}$,
%	 $P_n(\mathbf{X}_{h})$ denote $\begin{pmatrix} P_n(X_{T+1})\\ P_n(X_{T+2})\\ \vdots\\ P_n(X_{T+h})\end{pmatrix}$,
%	
%	 $Err(\mathbf{X}_{h})$ denote $\begin{pmatrix} Err(X_{T+1})\\ Err(X_{T+2})\\ \vdots\\ Err(X_{T+h})\end{pmatrix}$,
%	 $\mathbf{Y}_{h}$ denote $\begin{pmatrix} Y_{T+1}\\ Y_{T+2}\\ \vdots\\ Y_{T+h}\end{pmatrix}$,
%	  $P_n(\mathbf{Y}_{h})$ denote  $\begin{pmatrix} P_n(Y_{T+1})\\ P_n(Y_{T+2})\\ \vdots\\ P_n(Y_{T+h})\end{pmatrix}$,
%	  $Err(\mathbf{Y}_{h})$ denote  $\begin{pmatrix} Err(Y_{T+1})\\ Err(Y_{T+2})\\ \vdots\\ Err(X_{Y+h})\end{pmatrix}$
%%

Let $\mathbf{X}_{h}$ denote $[X_{n + i}]_{i = 1}^h$,
$P_n\mathbf{X_h}$ denote $[P_nX_{n+i}]_{i=1}^h$,
$Err(\mathbf{X}_h)$ denote $[Err(X_{n+i})]_{i=1}^h$,
	$\mathbf{Y}_{h}$ denote  $[Y_{n + i}]_{i = 1}^h$,
	$Err(\mathbf{Y}_h)$ denote $[Err(Y_{n+i})]_{i=1}^h$

Error is defined as
\begin{equation}
Err(Y_{n+h})=P_nY_{n+h}-Y_{n+h}
\end{equation}
From lemma \ref{C1} and \ref{C3C2}
\begin{equation}\label{PminersX}
Err(\mathbf{Y}_h)=C_3C_2\begin{pmatrix}
Z_{p-q+1}\\
Z_{p-q+2}\\
\vdots\\
Z_{n}
\end{pmatrix}
-C_1\begin{pmatrix}
Z_{n-q+1}\\
Z_{n-q+2}\\
\vdots\\
Z_{n+h}
\end{pmatrix}
\end{equation}
After augment the coefficient matrix $C_1$ and $C_2$ to $C_1^*$ and $C_2^*$ by
\begin{equation}
\begin{array}{c}
C_1^*=\left(O_{h\times (n-p)} \quad C_1\right)\\
C_2^*=\left(C_2 \quad O_{(n-p)\times h}\right)
\end{array}
\end{equation}
equation \ref{PminersX} is transformed to
\begin{equation}\label{Z2Error}
\begin{split}
Err(\mathbf{Y}_h)=&C_3C_2^*\begin{pmatrix}
Z_{p-q+1}\\
Z_{p-q+2}\\
\vdots\\
Z_{n+h}
\end{pmatrix}
-C_1^*\begin{pmatrix}
Z_{p-q+1}\\
Z_{p-q+2}\\
\vdots\\
Z_{n+h}
\end{pmatrix}\\
=&(C_3C_2^*-C_1^*)\begin{pmatrix}
Z_{p-q+1}\\
Z_{p-q+2}\\
\vdots\\
Z_{n+h}
\end{pmatrix}
\end{split}
\end{equation}
Let $C_{ZtoErrY}$ denote $C_3C_2^*-C_1^*$. Follow Lemma \ref{lineartransform},  $Err(\mathbf{Y}_h)\sim N_h(0,\Sigma_{ErrY})$, where
\begin{equation}
\Sigma_{ErrY}=C_{ZtoErrY}\Sigma_ZC_{ZtoErrY}^\top
\end{equation}
where the size of $\Sigma_Z$ here is $(n-p+q)\times(n-p+q)$.
\end{proof}

\begin{theorem}\label{stationaryconclusion}
$Err(\mathbf{X}_h)$ is a multivariate normal distribution $Err(\mathbf{X}_h)\sim N_h(0,\Sigma_{ErrX})$ in ARMA(p,q) process $\{X_t\}$, using the best linear predictor $P_nX_{n+h}$.
The covariance matrix of prediction error is
\begin{equation}
\Sigma_{ErrX} = C_{ErrYtoErrX}C_{ZtoErrY}\Sigma_ZC_{ZtoErrY}^\top C_{ErrYtoErrX}^\top
\end{equation}
	
	 where
\begin{equation}
C_{ErrYtoErrX}=\begin{pmatrix}
1 & 0 & 0 & \dots & 0\\
c_{21} & 1 & 0 & \dots & 0 \\
\vdots & \vdots & \ddots & \ddots & \vdots \\
c_{h1} & c_{h2} & c_{h3} & \dots & 1
\end{pmatrix}
\end{equation}
in which $c_{ij}$ is the coefficient of $\sigma Err(Y_{n+i})$ when calculating $Err(X_{n+j})$.

$c_{ij}$ can be recursively given by:
For any ($i \geq 2$),
\begin{equation}
\begin{split}
c_{i,1}=\sum_{k=1}^{min(p,h-1)}\phi_kc_{i-k,1}
\end{split}
\end{equation}

For any ($j \geq 2$),
\begin{equation}
c_{i,j}=c_{i-1,j-1}
\end{equation}

\end{theorem}

\begin{proof}
From ARIMA we know
\begin{equation}\label{Ytdefinition}
Y_t=\phi(B)X_t \quad t>max(p,q)
\end{equation}
Then from definition \ref{Ytdefinition}
\begin{equation}\label{definitionXt}
X_{n+h}=Y_{n+h}+\sum_{i=1}^p\phi_iX_{n+h-i}
\end{equation}

From definition \ref{PredictorPropertyAdd}
\begin{equation}
\begin{split}
P_nX_{n+h}&=P_n(Y_{n+h}+\sum_{i=1}^p\phi_iX_{n+h-i})\\
			&=P_nY_{n+h}+\sum_{i=1}^p\phi_iP_nX_{n+h-i}
\end{split}
\end{equation}

From corollary \ref{ObservedEReal}
\begin{equation}
P_nX_{n+h}=\left\{
\begin{array}{lcl}
P_nY_{n+h}+\sum\limits_{i=1}^{h-1}\phi_iP_nX_{n+h-i}+\sum\limits_{i=h}^p\phi_iX_{n+h-i}&&{h\leq p}\\
&&\\
P_nY_{n+h}+\sum\limits_{i=1}^{p}\phi_iP_nX_{n+h-i}&&{h>p}
\end{array}
\right.
\end{equation}

Then we have
\begin{equation}
\begin{split}
Err(X_{n+h})&=P_nX_{n+h}-X_{n+h}\\
&=P_nY_{n+h}+\sum\limits_{i=1}^{h-1}\phi_iP_nX_{n+h-i}+\sum\limits_{i=h}^p\phi_iX_{n+h-i}\\
&-Y_{n+h}-\sum_{i=1}^p\phi_iX_{n+h-i}\\
=&Err(Y_{n+h})+\sum\limits_{i=1}^{min(p,(h-1))}\phi_iErr(X_{n+h-i})
\end{split}
\end{equation}

which can be represent in matrix form
\begin{equation}
Err(\mathbf{X}_h)=C_{ErrYtoErrX}Err(\mathbf{Y}_h)
\end{equation}
Let $C_{ErrYtoErrX}$ denote
\begin{equation}
\begin{pmatrix}
1 & 0 & 0 & \dots & 0 \\
c_{21} & 1 & 0 & \dots & 0 \\
\vdots & \vdots & \vdots & \dots & 0\\
c_{h1} & c_{h2} & c_{h3} & \dots & 1
\end{pmatrix}
\end{equation}

Using Lemma \ref{lineartransform}, we have $Err(\mathbf{X}_h)\sim N_h(0,\Sigma_ErrX)$ where
\begin{equation}
\begin{split}
\Sigma_{ErrX}&=C_{ErrYtoErrX}\Sigma_{ErrY}C_{ErrYtoErrX}^\top\\
=&C_{ErrYtoErrX}C_{ZtoErrY}\Sigma_ZC_{ZtoErrY}^\top C_{ErrYtoErrX}^\top
\end{split}
\end{equation}

Q.E.D.
\end{proof}

From definition \ref{ErrorDef},
\begin{equation}
X_{n+h}\in [c_1,c_2] \Leftrightarrow Err(X_{n+h})\in [P_nX_{n+h}-c_2,P_nX_{n+h}-c_1]
\end{equation}
which lead to the following conclusion:

\begin{theorem}
\begin{equation*}
\begin{split}
&P(X_{n+h}\in [c_1,c_2])\\
=&P(Err(X_{n+h})\in [P_nX_{n+h}-c_2,P_nX_{n+h}-c_1])
\end{split}
\end{equation*}
\end{theorem}

Now we have a guarantee for any prediction interval of a time step $n+h$, denote the event $(X_{n+h}\in [c_{h,1},c_{h,2}])$ as event $A_h$, then the joint probability of $(\bigcup_{i=1}^hA_i)$ can be calculated by using theorem \ref{stationaryconclusion} when $\{X_t\}$ is a stationary process.

\subsection{Error of Intrinsically Stationary Process}
\begin{definition}
When the d-ordered differencing of a time series $\{X_t\}$ is a stationary process while its $(d-1)$-ordered process is still a non-stationary process, we call this process a d-ordered intrinsically stationary process.

A d-ordered differencing is defined as:
\begin{equation}\label{defferencingdefinition}
\nabla^dX_t=\sum_{k=0}^d\begin{pmatrix} d \\ k \end{pmatrix}(-1)^kX_{t-k}
\end{equation}
\end{definition}

\begin{lemma}\label{prediction}
To a d-ordered intrinsically stationary process, the best linear prediction of $X_{n+h}$ based on observation $\{1,X_1,X_2,\dots,X_n\}$ is:
\begin{equation}
P_nX_{n+h}=\left\{
\begin{array}{rcl}
\begin{split}
&P_n\nabla^dX_{n+h}-\sum\limits_{k=h}^d\begin{pmatrix} d \\ k \end{pmatrix}(-1)^kX_{n+h-k}\\
&-\sum\limits_{k=1}^{h-1}\begin{pmatrix} d \\ k \end{pmatrix}(-1)^kP_nX_{n+h-k}
\end{split}
& & {h\leq d+1}\\
 & & \\
P_n\nabla^dX_{n+h}-\sum\limits_{k=1}^{d}\begin{pmatrix} d \\ k \end{pmatrix}(-1)^kP_nX_{n+h-k} & & {h>d+1}
\end{array}
\right.
\end{equation}
\end{lemma}

\begin{proof}
From the definition of differencing function, equation \ref{defferencingdefinition},
\begin{equation}
X_{n+h}=\nabla^dX_{n+h}-\sum_{k=1}^d\begin{pmatrix} d \\ k \end{pmatrix}(-1)^kX_{n+h-k}
\end{equation}
which means
\begin{equation}
P_nX_{n+h}=P_n(\nabla^dX_{n+h}-\sum_{k=1}^d\begin{pmatrix} d \\ k \end{pmatrix}(-1)^kX_{n+h-k})
\end{equation}
using Property \ref{PredictorPropertyAdd}, %%%%(cite introduction to time series and forecast page59)
\begin{equation}
P_nX_{n+h}=P_n\nabla^dX_{n+h}-\sum_{k=1}^d\begin{pmatrix} d \\ k \end{pmatrix}(-1)^kP_nX_{n+h-k}
\end{equation}

Introducing Property \ref{ObservedEReal}, when $h\leq d+1$, $\{X_n,X_{n-1},\dots,X_{n+h-d}\}$ is observed, so:
\begin{equation}
\begin{split}
&P_nX_{n+h}\\
=&P_n\nabla^dX_{n+h}-\sum_{k=1}^d\begin{pmatrix} d \\ k \end{pmatrix}(-1)^kP_nX_{n+h-k}\\
=&P_n\nabla^dX_{n+h}-\sum_{k=h}^d\begin{pmatrix} d \\ k \end{pmatrix}(-1)^kX_{n+h-k}-\sum_{k=1}^{h-1}\begin{pmatrix} d \\ k \end{pmatrix}(-1)^kP_nX_{n+h-k}
\end{split}
\end{equation}

When $h>d+1$,
\begin{equation}
P_nX_{n+h}=P_n\nabla^dX_{n+h}-\sum_{k=1}^d\begin{pmatrix} d \\ k \end{pmatrix}(-1)^kP_nX_{n+h-k}
\end{equation}

Q.E.D.
\end{proof}

\begin{lemma}\label{ErrnablaXtoErrX}
The error of the $T+h$-th time step $Err(X_{n+h})$ can be represent recursively by $Err(\nabla^d\mathbf{X}_h)$.
\begin{equation}
Err(X_{n+h})=Err(\nabla^dX_{n+h})-\sum_{k=1}^{min(h-1,d)}\begin{pmatrix} d \\ k \end{pmatrix}(-1)^kErr(X_{n+h-k})
\end{equation}
\end{lemma}

\begin{proof}
When, $h\leq d+1$, using the definition of differencing (equation \ref{defferencingdefinition}):
\begin{equation}\label{h=hdefinition}
X_{n+h}=\nabla^dX_{n+h}-\sum_{k=h}^d\begin{pmatrix} d \\ k \end{pmatrix}(-1)^kX_{n+h-k}-\sum_{k=1}^{h-1}\begin{pmatrix} d \\ k \end{pmatrix}(-1)^kX_{n+h-k}
\end{equation}
and Lemma \ref{prediction}, the error $Err(X_{n+h})$ is:
\begin{equation}
\begin{split}
Err(X_{n+h})=&P_nX_{n+h}-X_{n+h}\\
=&P_n\nabla^dX_{n+h}-\sum_{k=h}^d\begin{pmatrix} d \\ k \end{pmatrix}(-1)^kX_{n+h-k}\\
&-\sum_{k=1}^{h-1}\begin{pmatrix} d \\ k \end{pmatrix}(-1)^kP_nX_{n+h-k}-X_{n+h}\\
=&P_n\nabla^dX_{n+h}-\sum_{k=h}^d\begin{pmatrix} d \\ k \end{pmatrix}(-1)^kX_{n+h-k}\\
&-\sum_{k=1}^{h-1}\begin{pmatrix} d \\ k \end{pmatrix}(-1)^kP_nX_{n+h-k}-\nabla^dX_{n+h}\\
&+\sum_{k=h}^d\begin{pmatrix} d \\ k \end{pmatrix}(-1)^kX_{n+h-k}+\sum_{k=1}^{h-1}\begin{pmatrix} d \\ k \end{pmatrix}(-1)^kX_{n+h-k}\\
=&Err(\nabla^dX_{n+h})-\sum_{k=1}^{h-1}\begin{pmatrix} d \\ k \end{pmatrix}(-1)^kErr(X_{n+h-k})
\end{split}
\end{equation}

Then for $h > d+1$, similarly, using the definition (equation \ref{defferencingdefinition}):
\begin{equation}\label{h>d+1definition}
X_{n+h}=\nabla^dX_{n+h}-\sum_{k=1}^{d}\begin{pmatrix} d \\ k \end{pmatrix}(-1)^kX_{n+h-k}
\end{equation}
and Theorem \ref{prediction}, the error $Err(X_{n+h})$ is:
\begin{equation}
\begin{split}
Err(X_{n+h})=&P_X{n+h}-X_{n+h}\\
=&P_n\nabla^dX_{n+h}-\sum_{k=1}^{d}\begin{pmatrix} d \\ k \end{pmatrix}(-1)^kP_nX_{n+h-k}-X_{n+h}\\
=&P_n\nabla^dX_{n+h}-\sum_{k=1}^{d}\begin{pmatrix} d \\ k \end{pmatrix}(-1)^kP_nX_{n+h-k}\\
&-\nabla^dX_{n+h}+\sum_{k=1}^{d}\begin{pmatrix} d \\ k \end{pmatrix}(-1)^kX_{n+h-k}\\
=&Err(\nabla^dX_{n+h})-\sum_{k=1}^{d}\begin{pmatrix} d \\ k \end{pmatrix}(-1)^kErr(X_{n+h-k})
\end{split}
\end{equation}
Q.E.D.
\end{proof}

\begin{theorem}\label{nonstationaryconclusion}
The joint distribution among errors of different time steps is a multivariate normal distribution
\begin{equation}
Err(\mathbf{X}_h)\sim N_h(0,\Sigma_ErrX)
\end{equation}
where
\begin{equation}
\Sigma_ErrX=C_{Err\nabla XtoErrX}\Sigma_{Err\nabla^dX}C_{Err\nabla XtoErrX}^\top
\end{equation}
in which $\Sigma_{Err\nabla^dX}$ is the covariance matrix of $Err(\nabla^d\mathbf{X}_h)$ given by Theorem \ref{stationaryconclusion} and
\begin{equation}
C_{Err\nabla XtoErrX}=\begin{pmatrix}
1 & 0 & 0 & \dots & 0 \\
T_{21} & 1 & 0 & \dots & 0 \\
\vdots & \vdots & \vdots & \dots & 0\\
T_{h1} & T_{h2} & T_{h3} & \dots & 1
\end{pmatrix}
\end{equation}
in which $T_{ij}$ is the coefficient of $Err(\nabla^dX_{T+i})$ when calculating $Err(X_{T+j})$.

$T_{ij}$ can be recursively given by:
For any ($i \geq 2$),
\begin{equation}
\begin{split}
T_{i,1}=\sum_{k=1}^{min(h-1,d)}\begin{pmatrix} d \\ k \end{pmatrix}(-1)^{k+1}T_{i-k,1}
\end{split}
\end{equation}

For any ($j \geq 2$),
\begin{equation}
T_{i,j}=T_{i-1,j-1}
\end{equation}
\end{theorem}

\begin{proof}
From Lemma \ref{ErrnablaXtoErrX}
\begin{equation}
\begin{pmatrix}
Err(X_{n+1}) \\
Err(X_{n+2}) \\
\vdots \\
Err(X_{n+h}) \end{pmatrix}
=\begin{pmatrix}
1 & 0 & 0 & \dots & 0 \\
T_{21} & 1 & 0 & \dots & 0 \\
\vdots & \vdots & \vdots & \dots & 0\\
T_{h1} & T_{h2} & T_{h3} & \dots & 1
\end{pmatrix}
\begin{pmatrix}
Err(\nabla^dX_{n+1}) \\ Err(\nabla^dX_{n+2}) \\ \vdots \\ Err(\nabla^dX_{n+h})
\end{pmatrix}
\end{equation}
in which $T_{ij}$ is the coefficient of $Err(\nabla^dX_{n+i})$ when calculating $Err(X_{n+j})$.

$T_{ij}$ can be recursively given by:
For any ($i \geq 2$),
\begin{equation}
\begin{split}
T_{i,1}=\sum_{k=1}^{min(h-1,d)}\begin{pmatrix} d \\ k \end{pmatrix}(-1)^{k+1}T_{i-k,1}
\end{split}
\end{equation}

For any ($j \geq 2$),
\begin{equation}
T_{i,j}=T_{i-1,j-1}
\end{equation}

Denoting
\begin{equation*}
Err(\nabla^d\mathbf{X}_h):=\begin{pmatrix}
Err(\nabla^dX_{n+1})\\
Err(\nabla^dX_{n+2})\\
\vdots\\
Err(\nabla^dX_{n+h})\end{pmatrix}
\end{equation*}
\begin{equation*}
C_{Err\nabla XtoErrX}:=\begin{pmatrix}
1 & 0 & 0 & \dots & 0 \\
T_{21} & 1 & 0 & \dots & 0 \\
\vdots & \vdots & \vdots & \dots & 0\\
T_{h1} & T_{h2} & T_{h3} & \dots & 1
\end{pmatrix}
\end{equation*}
\begin{equation*}
Err(\mathbf{X}_h):=\begin{pmatrix}
Err(X_{n+1}) \\
Err(X_{n+2}) \\
\vdots \\
Err(X_{n+h}) \end{pmatrix}
\end{equation*}

Following Definition \ref{lineartransform}, $Err(\mathbb{X}_h)\sim N_h(0,\Sigma_{ErrX})$ where
\begin{equation}
\Sigma_{ErrX}=
C_{Err\nabla XtoErrX}\Sigma_{Err\nabla X}C_{Err\nabla XtoErrX}^\top
\end{equation}

Q.E.D.
\end{proof}

Now we have the joint guarantee of any prediction interval of a series of time steps $\{X_{n+1},X_{n+2},\dots ,X_{n+h}\}$.We can use the conclusion of Theorem \ref{nonstationaryconclusion} when dealing with an intrinsically stationary process.

\section{Error of Prediction of Multivariate Process}
\subsection{Error of Stationary Multivariate Process}
All of the time series we have discussed are univariate time series, now we want to generalize our conclusion to multivariate cases. Similarly, we will deal with the stationary case firstly.
\begin{definition}
The best linear prediction of $X_{n+h}$ is\cite{StatTexBook}:
\begin{equation}
P(X_{n+h}\mid X_n,X_{n-1},\dots ,X_1,\begin{pmatrix} 1 \\ 1 \\ \vdots \\ 1\end{pmatrix})=A_0^h+\sum_{i=1}^nA_i^hX_{n+1-i}
\end{equation}
in which $\{X_t\}$ is a time series with m variate, $A_0^h$ is a $m\times 1$ vector and $A_i^h$ are $n\times n$ matrices.

The optimized choice of $A_i^h$ are
\begin{equation}\label{multibestchoice}
\sum_{j=1}^nA_j^h\Gamma(i-j)=\Gamma(i+h-1)\quad {i=1,2,\dots,n}
\end{equation}
\end{definition}

\begin{definition}
Denoting $\mathbf{Z}_t:=[Z_i]_{i=1}^h$, where $Z_i$ is the $[m\times m]$ white noise of time step i.
\begin{equation}
\mathbf{Z}_t\sim N_{mt}(0,\Sigma_{\mathbf{Z}})
\end{equation}
where
\begin{equation}
\Sigma_{\mathbf{Z}}=\sigma^2E_{mt\times mt}
\end{equation}
\end{definition}

\begin{lemma}\label{multiZtoY}
Defining a multivariate MA process $\{Y_t\}$
\begin{equation}\label{multiYtdefinition}
Y_t=\Phi(B)X_t \quad t>max(p,q)
\end{equation}
where $\Phi_i$ is a $m\times m$ matrix.

Denoting $\mathbf{Y}_h:=[Y_{n+i}]_{i=1}^h$, then
\begin{equation}
\mathbf{Y}_h=C_1\begin{pmatrix}
Z_{n-q+1}\\
Z_{n-q+2}\\
\vdots\\
Z_{n+h}\end{pmatrix}
\end{equation}
where $C_1$ is a $[mh\times m(q+h)]$ transform matrix from
$\begin{pmatrix}
Z_{n-q+1}\\
Z_{n-q+2}\\
\vdots\\
Z_{n+h}\end{pmatrix}$
to $\mathbf{Y}_h$.
\begin{equation}\label{multiZ2Ytransformmatrix}
C_1=\begin{pmatrix}
\Theta_q & \Theta_{q-1} & \dots & \Theta_1 & E_{m\times m} & O_{m\times m} & O_{m\times m} & \dots & O_{m\times m}\\
O_{m\times m} & \Theta_q & \Theta_{q-1} & \dots & \Theta_1 & E_{m\times m} & O_{m\times m} & \dots & O_{m\times m}\\
O_{m\times m} & O_{m\times m} & \Theta_q & \Theta_{q-1} & \dots & \Theta_1 & E & \dots & O_{m\times m}\\
\vdots & \vdots & \vdots & \ddots & \ddots & \ddots & \ddots & \ddots & \vdots\\
O_{m\times m} & O_{m\times m} & O_{m\times m} & \dots & \Theta_q & \Theta_{q-1} & \dots & \Theta_1 & E_{m\times m}
\end{pmatrix}
\end{equation}
\end{lemma}

\begin{proof}
From the definition of multivariate ARMA(p,q) process,
\begin{equation}
\begin{split}
Y_t=&\Theta(B)Z_t\quad {t\geq p+1}\\
=&Y_t=Z_t+\Theta_1Z_{t-1}+\Theta_2Z_{t-2}+\dots +\Theta_qZ_{t-q}
\end{split}
\end{equation}

Q.E.D.
\end{proof}

\begin{lemma}\label{multiZtoPY}
The best linear predictor of $\mathbf{Y}_h$ is a linear combination of
$\begin{pmatrix}
Z_{p-q+1}\\
Z_{p-q+2}\\
\vdots\\
Z_{n}
\end{pmatrix}$.
\begin{equation}
P_n\mathbb{Y}_h=C_3C_2\begin{pmatrix}
Z_{p-q+1}\\
Z_{p-q+2}\\
\vdots\\
Z_{n}
\end{pmatrix}
\end{equation}
where
\begin{equation}
C_2=\begin{pmatrix}
\Theta_q & \Theta_{q-1} & \dots & \Theta_1 & E_{m\times m} & O_{m\times m} & O_{m\times m} & \dots & O_{m\times m}\\
O_{m\times m} & \Theta_q & \Theta_{q-1} & \dots & \Theta_1 & E_{m\times m} & O_{m\times m} & \dots & O_{m\times m}\\
O_{m\times m} & O_{m\times m} & \Theta_q & \Theta_{q-1} & \dots & \Theta_1 & E_{m\times m} & \dots & O_{m\times m}\\
\vdots & \vdots & \vdots & \ddots & \ddots & \ddots & \ddots & \ddots & \vdots\\
O_{m\times m} & O_{m\times m} & O_{m\times m} & \dots & \Theta_q & \Theta_{q-1} & \dots & \Theta_1 & E_{m\times m}
\end{pmatrix}
\end{equation}
is a $[m(n-p)\times m(n-p+q)]$ transform matrix from
$
\begin{pmatrix}
Z_{p-q+1}\\
Z_{p-q+2}\\
\vdots\\
Z_{n}
\end{pmatrix}
$ to
$
\begin{pmatrix}
Y_{p+1}\\
Y_{p+2}\\
\vdots\\
Y_{n}\end{pmatrix}
$.

And
\begin{equation}
C_3=\begin{pmatrix}
A_{n-p}^1 & A_{n-p-1}^1 & \dots & A_1^1 \\
A_{n-p}^2 & A_{n-p-1}^2 & \dots & A_1^2 \\
\vdots & \vdots & \ddots & \vdots \\
A_{n-p}^h & A_{n-p-1}^h & \dots & A_1^h
\end{pmatrix}
\end{equation}
is the coefficient matrix of $(Y_{p+1},Y_(p+2),\dots,Y_n)$ when calculating $P_T\mathbb{Y}_h$, the best linear predictor of $\{Y_{n+1},Y_{n+2},\dots,Y_{n+h}\}$, where $\{Y_t\}$ a multivariate stationary MA(q) process with zero mean.

$A_{i}^j$ will be given by equation \ref{multibestchoice}.
\end{lemma}

\begin{theorem}
The joint distribution among errors of different time steps of $\mathbf{Y}_h$ is a multivariate distribution
$$Err(\mathbf{Y}_h)\sim N_{mh}(0,\Sigma_{ErrY})$$, where
\begin{equation}
\begin{split}
\Sigma_{ErrY}=&C_{ZtoErrY}\Sigma_ZC_{ZtoErrY}^\top\\
=&\sigma^2C_{ZtoErrY}E_{m(T+h-p+q)\times m(T+h-p+q)}C_{ZtoErrY}^\top
\end{split}
\end{equation}
\end{theorem}

\begin{proof}
Denoting $\mathbf{X}_h:=[X_{n+i}]_{i=1}^h$, $P_n\mathbf{X}_h:=[P_nX_{n+i}]_{i=1}^h$, $Err(\mathbf{X}_h):=[Err(X_{n+i})]_{i=1}^h$, $\mathbf{Y}_h:=[Y_{n+i}]_{i=1}^h$, $P_n\mathbf{Y}_h:=[P_nY_{n+i}]_{i=1}^h$, $Err(\mathbf{Y}_h):=[Err(Y_{n+i})]_{i=1}^h$,

Error of $Y_{T+h}$ is defined as
\begin{equation}
Err(Y_{n+h})=P_nY_{n+h}-Y_{n+h}
\end{equation}
From Lemma \ref{multiZtoY} and Lemma \ref{multiZtoPY},
\begin{equation}\label{multiPminersX}
Err(\mathbb{Y}_h)=C_3C_2\begin{pmatrix}
Z_{p-q+1}\\
Z_{p-q+2}\\
\vdots\\
Z_{n}
\end{pmatrix}
-C_1\begin{pmatrix}
Z_{n-q+1}\\
Z_{n-q+2}\\
\vdots\\
Z_{n+h}
\end{pmatrix}
\end{equation}

Augment the coefficient matrices  $C_1$ and $C_2$ to $C_1^*$ and $C_2^*$ by
\begin{equation}
\begin{array}{c}
C_1^*=\left(O_{mh\times m(n-p)} \quad C_1\right)\\
C_2^*=\left(C_2 \quad O_{m(n-p)\times mh}\right)
\end{array}
\end{equation}

Equation \ref{multiPminersX} can be transformed to
\begin{equation}\label{multiZ2Error}
\begin{split}
Err(\mathbb{Y}_h)=&C_3C_2^*\begin{pmatrix}
Z_{p-q+1}\\
Z_{p-q+2}\\
\vdots\\
Z_{n+h}
\end{pmatrix}
-C_1^*\begin{pmatrix}
Z_{p-q+1}\\
Z_{p-q+2}\\
\vdots\\
Z_{n+h}
\end{pmatrix}\\
=&(C_3C_2^*-C_1^*)\begin{pmatrix}
Z_{p-q+1}\\
Z_{p-q+2}\\
\vdots\\
Z_{n+h}
\end{pmatrix}
\end{split}
\end{equation}

Denoting $C_{ZtoErrY}:=C_3C_2^*-C_1^*$. Follow Definition \ref{lineartransform}, $Err(\mathbb{Y}_h)\sim N_{mh}(0,\Sigma_{ErrY})$, where
\begin{equation}
\Sigma_{ErrY}=C_{ZtoErrY}\Sigma_ZC_{ZtoErrY}^\top
\end{equation}

Q.E.D.
\end{proof}

\begin{theorem}\label{multivariatestationary}
When $\{X_t\}$ is a multivariate ARMA(p,q) process with n variables, using the best linear predictor $P_nX_{n+h}$, the joint distribution among errors of different time steps is a multivariate normal distribution
$$Err(\mathbb{X}_{h})\sim N_{mh}(0,\Sigma_{ErrX})$$,
\begin{equation}
\begin{split}
\Sigma_{ErrX}=&C_{ErrYtoErrX}C_{ZtoErrY}\Sigma_ZC_{ZtoErrY}^\top C_{ErrYtoErrX}^\top\\
=&\sigma^2C_{ErrYtoErrX}C_{ZtoErrY}E_{m(n+h-p+q)\times m(n+h-p+q)}C_{ZtoErrY}^\top C_{ErrYtoErrX}^\top
\end{split}
\end{equation}
where
\begin{equation}
C_{ErrYtoErrX}=\begin{pmatrix}
E_{m\times m} & O_{m\times m} & O_{m\times m} & \dots & O_{m\times m}\\
C_{21} & E_{m\times m} & O_{m\times m} & \dots & O_{m\times m} \\
\vdots & \vdots & \ddots & \ddots & \vdots \\
C_{h1} & C_{h2} & C_{h3} & \dots & E_{m\times m}
\end{pmatrix}
\end{equation}
in which $C_{ij}$ is the coefficient matrix of $\sigma Err(Y_{n+i})$ when calculating $Err(X_{n+j})$.

$C_{ij}$ can be recursively given by:
For any ($i \geq 2$),
\begin{equation}
\begin{split}
C_{i,1}=\sum_{k=1}^{min(p,h-1)}\Phi_kC_{i-k,1}
\end{split}
\end{equation}

For any ($j \geq 2$),
\begin{equation}
C_{i,j}=C_{i-1,j-1}
\end{equation}
\end{theorem}

\begin{proof}
From the definition of $Y_t$, equation \ref{multiYtdefinition}, we have
\begin{equation}\label{multidefinitionXt}
X_{n+h}=Y_{n+h}+\sum_{i=1}^p\Phi_iX_{n+h-i}
\end{equation}

The best linear prediction of $X_{n+h}$
\begin{equation}
P_nX_{n+h}=P_n(Y_{n+h}+\sum_{i=1}^p\Phi_iX_{n+h-i})
\end{equation}

From Property \ref{PredictorPropertyAdd} we have:
\begin{equation}
P_nX_{n+h}=P_nY_{n+h}+\sum_{i=1}^p\Phi_iP_nX_{n+h-i}
\end{equation}

Introducing Property \ref{ObservedEReal}:
\begin{equation}
P_nX_{n+h}=\left\{
\begin{array}{lcl}
P_nY_{n+h}+\sum\limits_{i=1}^{h-1}\Phi_iP_nX_{n+h-i}+\sum\limits_{i=h}^p\Phi_iX_{n+h-i}&&{h<p}\\
&&\\
P_nY_{n+h}+\sum\limits_{i=1}^{p}\Phi_iP_nX_{n+h-i}&&{h\leq p}
\end{array}
\right.
\end{equation}

According to equation \ref{multidefinitionXt},
\begin{equation}
\begin{split}
Err(X_{n+h})&=P_nX_{n+h}-X_{n+h}\\
&=P_nY_{n+h}+\sum\limits_{i=1}^{h-1}\Phi_iP_nX_{n+h-i}+\sum\limits_{i=h}^p\Phi_iX_{n+h-i}\\
&-Y_{n+h}-\sum_{i=1}^p\Phi_iX_{n+h-i}\\
=&Err(Y_{n+h})+\sum\limits_{i=1}^{min(p,(h-1))}\Phi_iErr(X_{n+h-i})
\end{split}
\end{equation}

which can be represent in matrix form
\begin{equation}
Err(\mathbb{X}_h)=\begin{pmatrix}
E_{m\times m} & O_{m\times m} & O_{m\times m} & \dots & O_{m\times m}\\
C_{21} & E_{m\times m} & O_{m\times m} & \dots & O_{m\times m} \\
\vdots & \vdots & \ddots & \ddots & \vdots \\
C_{h1} & C_{h2} & C_{h3} & \dots & E_{m\times m}
\end{pmatrix}
Err(\mathbb{Y}_h)
\end{equation}
in which $C_{ij}$ is the coefficient matrix of $\sigma Err(Y_{n+i})$ when calculating $Err(X_{n+j})$.

$C_{ij}$ can be recursively given by:
For any ($i \geq 2$),
\begin{equation}
\begin{split}
C_{i,1}=\sum_{k=1}^{min(p,h-1)}\Phi_kC_{i-k,1}
\end{split}
\end{equation}

For any ($j \geq 2$),
\begin{equation}
C_{i,j}=C_{i-1,j-1}
\end{equation}

Denoting
\begin{equation}
C_{ErrYtoErrX}:=\begin{pmatrix}
E_{m\times m} & O_{m\times m} & O_{m\times m} & \dots & O_{m\times m}\\
C_{21} & E_{m\times m} & O_{m\times m} & \dots & O_{m\times m} \\
\vdots & \vdots & \ddots & \ddots & \vdots \\
C_{h1} & C_{h2} & C_{h3} & \dots & E_{m\times m}
\end{pmatrix}
\end{equation}

Using Definition \ref{lineartransform}, $Err(\mathbb{X}_h)\sim N_{mh}(0,\Sigma_{ErrX})$, where
\begin{equation}
\begin{split}
\Sigma_{ErrX}&=C_{ErrYtoErrX}\Sigma_{ErrY}C_{ErrYtoErrX}^\top\\
=&C_{ErrYtoErrX}C_{ZtoErrY}\Sigma_ZC_{ZtoErrY}^\top C_{ErrYtoErrX}^\top
\end{split}
\end{equation}

Q.E.D.
\end{proof}

With Theorem \ref{multivariatestationary}, we can calculate the joint probability of any prediction interval of any certain series of time steps $\{X_{n+1}^s,X_{n+2}^s,\dots ,X_{n+h}^s\}$, where $\{X_t\}$ is a multivariate stationary process.

\subsection{Error of Multivariate Non-Stationary Process}
\begin{definition}
The definition of a d-ordered multivariate intrinsically stationary process is exactly the same as univariate one.And the differencing method is also:
\begin{equation}
\nabla^dX_t=\sum_{k=0}^d\begin{pmatrix} d \\ k \end{pmatrix}(-1)^kX_{t-k}
\end{equation}
\end{definition}

\begin{lemma}
For any certain series $\{X_t^s\}$, we have:
\begin{equation}
\nabla^dX_t^s=\sum_{k=0}^d\begin{pmatrix} d \\ k \end{pmatrix}(-1)^kX_{t-k}^s
\end{equation}
\end{lemma}

\begin{theorem}\label{multinonstationaryjointdistribution}
The joint distribution of $\begin{pmatrix} Err(X_{n+1}) \\ Err(X_{n+2}) \\ \vdots \\ Err(X_{n+h}) \end{pmatrix}$ is a multivariate normal distribution
$$Err(\mathbb{X}_h)\sim N_{mh}(0,\Sigma_{ErrX})$$, where
\begin{equation}
\Sigma_{ErrX}=
C_{Err\nabla XtoErrX}\Sigma_{\nabla h}C_{Err\nabla XtoErrX}^\top
\end{equation}
where $\Sigma_{\nabla h}$ is the covariance matrix of errors of the differenced stationary multivariate process given by Theorem \ref{multivariatestationary} and
\begin{equation}
C_{Err\nabla XtoErrX}=\begin{pmatrix}
E_{m\times m} & O_{m\times m} & O_{m\times m} & \dots & O_{m\times m} \\
T_{21}E_{m\times m} & E_{m\times m} & O_{m\times m} & \dots & O_{m\times m} \\
\vdots & \vdots & \vdots & \dots & O\\
T_{h1}E_{m\times m} & T_{h2}E_{m\times m} & T_{h3}E_{m\times m} & \dots & E_{m\times m}
\end{pmatrix}
\end{equation}
in which $T_{ij}$ can be recursively given by:
For any ($i \geq 2$),
\begin{equation}
\begin{split}
T_{i,1}=\sum_{k=1}^{h-1}\begin{pmatrix} d \\ k \end{pmatrix}(-1)^{k+1}T_{i-k,1}
\end{split}
\end{equation}

For any ($j \geq 2$),
\begin{equation}
T_{i,j}=T_{i-1,j-1}
\end{equation}
\end{theorem}

Now we can calculate the joint guarantee of any prediction interval of a certain series of time steps $\{X_{n+1}^s,X_{n+2}^s,\dots ,X_{n+h}^s\}$, when
$\{X_t\}$ is a multivariate intrinsically stationary process.\section{Conclusion}
We use the property of both prediction and modeling of ground truth value are linear combination of white noise, which is the assumption of ARIMA model, to prove that the prediction error: difference between ground truth value and predicted value, are also linear combination of white noise, thus prediction errors among different time steps is eligible to form a multivariate distribution and the according expression for different scenarios is in Theorem \ref{stationaryconclusion}, \ref{nonstationaryconclusion}, \ref{multivariatestationary}, \ref{multinonstationaryjointdistribution}.

\bibliographystyle{plain}
\bibliography{main.bib}{}

\end{document}